\documentclass[journal]{IEEEtran}

\usepackage{multicol}
\usepackage{multirow}
\usepackage{longtable}
\usepackage{subcaption}

\usepackage{balance}
\usepackage{cite}

\usepackage{algorithm}
\usepackage{algpseudocode}

\ifCLASSINFOpdf
   \usepackage[pdftex]{graphicx}
   \DeclareGraphicsExtensions{.pdf,.jpeg,.png}
\else
 
\fi

\usepackage{mathrsfs,amsmath}
\usepackage{amssymb}

\usepackage{amsthm,hyperref,enumerate}
\newtheorem{theorem}{Theorem}

\theoremstyle{definition}
\newtheorem{defn}{Definition}

\theoremstyle{remark}
\newtheorem{remark}{Remark}
\newtheorem{lemma}{Lemma}

\theoremstyle{plain}
 \newtheorem{assumption}{Assumption}

\usepackage{ragged2e}

\usepackage[amssymb]{SIunits}
\hypersetup{
	colorlinks=true,
	linkcolor=blue,
	filecolor=blue,      
	urlcolor=blue,
}

\begin{document}

\title{Distributed Formation Maneuver Control Using Complex Laplacian}

\author{Xu Fang,  Lihua Xie,~\IEEEmembership{Fellow,~IEEE}
\thanks{This work was supported 
by Nanyang Technological University under the Wallenberg-NTU Presidential Postdoctoral Fellowship. (Corresponding author: Lihua Xie.)}
\thanks{
The authors are with the School of Electrical and Electronic Engineering, Nanyang Technological University, Singapore. (E-mail: fa0001xu@e.ntu.edu.sg; 
elhxie@ntu.edu.sg).}
}

\maketitle

\begin{abstract}

This paper studies the problem of distributed formation maneuver control of multi-agent systems via complex Laplacian.  We will show how to change the translation, scaling, rotation, and also the shape of formation continuously by only tuning the positions of the leaders in both 2-D and 3-D spaces, where the rotation of formation in 3-D space is realized by changing the yaw angle, pitch angle, and roll angle of formation sequentially.   Compared with real-Laplacian-based methods, the first advantage of the proposed complex-Laplacian-based approach is that each follower requires fewer neighbors and lesser communication. The second advantage is that non-convex and non-generic nominal configurations are allowed and the uniqueness of the complex-constraint-based target formation can be guaranteed by the non-collocated 
nominal agents. The third advantage is that more formation shapes can be realized by only tuning the positions of the leaders. Two simulation examples are given to illustrate the theoretical results.
\end{abstract}


\IEEEpeerreviewmaketitle

\section{Introduction}

\IEEEPARstart{T}{here} has been a growing interest in control and localization of multi-agent systems due to their benefits in robot-related applications such as well-studied optimal control and intelligent home systems \cite{wang2019cooperative1, ameur2019intelligent, duan2022distributed}. 
The translation, scaling, rotation, and the shape of formation can be changed continuously by the formation maneuver control techniques.

The conventional consensus-based approaches can control the agents to form a time-varying formation, but each agent needs the information of its time-varying target position \cite{amirkhani2021consensus}. 
This article deals with the challenging case that only the leader group has information of the time-varying target formation and the follower group follows the leader group without the need of knowing their target positions. To tackle this challenging case, the existing results use equilibrium stresses \cite{zhao2018affine, chen2020distributed,  li2020layered, xu2020affine, onuoha2019 } or barycentric-coordinate constraints \cite{ han2017fobarycentric, han2015three, fang2021distributed} to describe the target formation, where the followers only need to keep the invariant geometric  relationship with their neighbors.
But each follower requires to have at least $d\!+\!1$ neighbors in $d$-dimensional space. In addition, they need generic \cite{zhao2018affine, chen2020distributed,  li2020layered, xu2020affine, onuoha2019,han2017fobarycentric}, convex \cite{han2015three}, or rigid \cite{fang2021distributed} nominal configurations to guarantee 
the uniqueness and localizability of their target formations. For instance, each nominal follower is inside the convex hull \cite{han2015three} or needs three non-colinear neighbors in 2-D plane or four non-coplanar neighbors in 3-D space \cite{zhao2018affine, chen2020distributed,  li2020layered, xu2020affine, onuoha2019,han2017fobarycentric,fang2021distributed}.

Different from the existing real-Laplacian-based formation maneuver control \cite{zhao2018affine, chen2020distributed, li2020layered, xu2020affine, onuoha2019,han2017fobarycentric,han2015three,fang2021distributed}, we will apply complex constraints to describe the target formation and realize formation maneuver control via complex Laplacian in both 2-D and 3-D spaces. Compared with existing real-Laplacian-based approaches \cite{zhao2018affine, chen2020distributed, li2020layered, xu2020affine, onuoha2019,han2017fobarycentric,han2015three,fang2021distributed}, 
the first advantage is that each follower requires fewer neighbors and lesser communication. Each follower only needs two neighbors to form a complex constraint in both 2-D and 3-D spaces, which also indicates that the proposed  approach requires lesser communication among the agents. The second advantage is that the nominal configuration can be non-convex and non-generic. We only require that each nominal follower and its neighbors are not collocated, which brings the third advantage that more formation shapes can be realized by only tuning the positions of the leaders.

The existing formation control methods with complex Laplacian \cite{ lin2014distributed,han2015formation,de2021distributed,fangsubmit} have some limitations. First, the existing works \cite{ lin2014distributed,han2015formation,de2021distributed,fangsubmit}  are limited to 2-D plane.  Second, the works \cite{lin2014distributed,han2015formation} can only tune the scale of the formation, while 
the works \cite{de2021distributed,fangsubmit} require that all agents know the time-varying maneuver parameters for constructing a time-varying complex Laplacian matrix. Third, the existing works \cite{ lin2014distributed,han2015formation,de2021distributed,fangsubmit}  do not study how to realize the time-varying formation shape by only tuning the positions of the leaders.
Fourth, the existing works
require undirected graphs \cite{ lin2014distributed,de2021distributed} or synchronized constant velocities of the leaders \cite{han2015formation}.
Different from the existing  results \cite{ lin2014distributed,han2015formation,de2021distributed,fangsubmit}, the proposed method under directed graphs can change the translation, scaling, rotation, and also the shape of formation continuously in both 2-D and 3-D spaces, where there is no constraint on the velocities of the leaders and the followers can achieve formation maneuver control without the need of knowing the time-varying maneuver parameters. 
The permitted formation shape change is determined by the design of the nominal formation.

The rest of the paper is organized as following. Section \ref{problem} introduces some preliminaries and control objectives. The control protocols for the agents in 2-D and 3-D spaces are given in 
Section \ref{2dp} and Section \ref{3dp}, respectively. In Section \ref{dis}, we explore the conditions under which the followers can be controlled to form any desired shape by only tuning the leaders' positions. The simulation examples and conclusion are given in Section \ref{lation} and Section \ref{conc}, respectively.

\section{Preliminaries and Problem Statement}\label{problem}

\subsection{Notations}

Let $\mathbb{C}$ and $\mathbb{R}$ denote the set of all the complex and real numbers, respectively. Let ${I}_d \in \mathbb{R}^{d \times d}$ be an identity matrix of dimension $d \times d$. Let ${\mathbf{0}}_d, \mathbf{1}_d$ be respectively the all-zero and all-one vectors of dimension $d$.  Denote $\mathcal{A}^T$ as the transpose of a real matrix $\mathcal{A} \! \in \! \mathbb{R}^{d \times d}$. Denote $\mathcal{A}^H$ as the conjugate transpose of a complex matrix $\mathcal{A} \! \in \! \mathbb{C}^{d \times d}$. Denote $\lambda_{\max}(\mathcal{A})$ and $\lambda_{\min}(\mathcal{A})$ as the maximum and minimum eigenvalues of the Hermitian  matrix $\mathcal{A} \! \in \! \mathbb{C}^{d \times d}$, respectively. Let $\otimes$ be the Kronecker product.
Let $\text{rank}(\cdot)$ be the rank of a matrix. Consider a multi-agent system of $n$ agents. The position of each agent $i$ is given by
\begin{equation}\label{complex}
    p_i = x_i + y_i \iota, \ i=1, \cdots, n, 
\end{equation}
where $x_i, y_i \in \mathbb{R}$ and
$\iota$ is the imaginary unit with $\iota^2 \!=\! -1$. Denote $p_i^H= x_i - y_i \iota$ as the complex conjugate of $p_i$. Denote $\| \cdot \|_2$ as the $L_2$ norm, e.g., $\| p_i \|_2 = \sqrt{x_i^2+y_i^2}$. 
The agents are connected over a network whose communication graph $\mathcal{G}=\{ \mathcal{V},\mathcal{E}\}$ consists of an agent set $\mathcal{V}=\{1,  \cdots, n \}$ and an edge set $\mathcal{E}  \subseteq \mathcal{V} \times \mathcal{V}$. If $(i,j) \in \mathcal{E}$, agent $i$ can obtain information from agent $j$. The set of neighbors
of agent $i$ is denoted by
$\mathcal{N}_i \triangleq \{ j \in \mathcal{V} : (i,j) \in \mathcal{E} \}$. Agent $i$ is said to be two-reachable from a set $\mathcal{V}_t \in \mathcal{V}$ if there are two disjoint paths from $\mathcal{V}_t$ to agent $i$ \cite{lin2014distributed}.

A formation of $n$ agents is represented by $(\mathcal{G}, p)$, where $p \!=\![p_1^{H},\cdots, p_n^{H}]^{H}$ is a configuration of the formation. 
In a complex plane,
a time-varying target formation of $n$ agents is represented by $(\mathcal{G}, p^*(t))$, where $p^*(t) \!=\![[p_1^*(t)]^{H},\cdots, [p_n^*(t)]^{H}]^{H}$. A constant nominal formation of $n$ agents is represented by $(\mathcal{G}, r)$, where $r \!=\! [r_1^{H},\cdots, r_n^{H}]^{H} \! \in \! \mathbb{C}^{n}$ is a nominal configuration. A nominal configuration is said to be generic if the coordinates $r_1, \cdots, r_n$ do not satisfy any nontrivial algebraic equation with integer coefficients. Intuitively speaking, for a generic configuration, the nominal agents are not collocated, 
any three nominal agents are not colinear in 2-D plane, and  any four nominal agents are not coplanar in 3-D space \cite{zhao2018affine, chen2020distributed,  li2020layered, xu2020affine, onuoha2019,han2017fobarycentric}.
Denote by $\mathcal{V}_L \!=\! \{1, \cdots, m \}$ and $\mathcal{V}_F\!=\! \{m\!+\!1, \cdots, n\}$ the leader group and follower group, respectively. Let 
$p_L^*(t) \!=\! [[p_1^*(t)]^H, \cdots, [p_{m}^*(t)]^H]^H $ and $p_F^*(t) \!=\! [[p_{m\!+\!1}^*(t)]^H , \cdots, [p_{n}^*(t)]^H]^H$ be the time-varying target positions of the leader group and follower group, respectively. Let 
$r_L \!=\! [r_1^H, \cdots, r_{m}^H]^H$ and $r_F \!=\! [r_{m\!+\!1}^H , \cdots, r_{n}^H]^H$ be the constant nominal positions of the leader group and follower group, respectively.

\subsection{Nonsmooth Analysis}

Consider a system $\dot{q}=f(q,t)$ ,
where $f(\cdot ):\mathbb{R}^{d}\times\mathbb{R}\rightarrow\mathbb{R}^{d}$ is a Lebesgue measurable and locally bounded discontinuous function \cite{paden1987calculus}. $q(\cdot)$ is called a Filippov solution of $f(q,t)$ on $[t_{0},t_{1}]$ if $q(\cdot)$ is absolutely continuous on $[t_{0},t_{1}]$ and for almost all $t\in \lbrack t_{0},t_{1}]$ satisfies
the differential inclusion $q\in \mathcal{K}[f](q,t)$ with $\mathcal{K}[f]:=\cap _{\varphi >0}\cap _{\mu (\bar{N})=0}\bar{co}(f(B(q,\varphi ) \!-\! \bar{N}),t)$. $B(q,\varphi )$  is the open ball of radius $\varphi$  centered at $q$.  $\bar{co}({N})$ is the convex closure of set $N$. $\cap _{\mu (\bar{N})=0}$ is the intersection over all sets $\bar{N}$ of Lebesgue measure zero. The convex hull is denoted by
$co(\cdot)$. The Clarke's generalized gradient of a locally Lipschitz continuous function $V$ is given by
$\partial V \! \triangleq \! co\{\lim \triangledown V(q_{i})|q_{i}\rightarrow q,q_{i}\in \Omega _{v}\cup \bar{N}\}$. $\Omega _{v}$ is the Lebesgue measure zero set where $\triangledown V(q_{i})$ does not exist.  $\bar{N}$ is a zero measure set \cite{paden1987calculus}. Then, the set-valued Lie derivative of $V$ with respect to $\dot{x}=f(q,t)$ is given by $\dot{\tilde{V}}:=\cap _{\phi \in\partial V}\phi ^{\top}\mathcal{K}[f](q,t)$.

\subsection{Complex Constraint}

The distance between nominal agent $i$ and nominal agent $j$ is denoted by $d_{ij} \!=\! \|r_j - r_i\|_2$.
In a constant nominal formation $(\mathcal{G}, r)$, if each nominal follower $i$ and its two neighbors $j,k$ are not collocated, i.e., $d_{ij}, d_{ik} \neq 0$, the position relationship among each nominal follower $i$ and its two neighbors $j,k$ can be described by a complex constraint, i.e.,
\begin{equation}\label{complexc}
    w_{ij}(r_j-r_i)+w_{ik}(r_k-r_i)= 0, \ i \in \mathcal{V}_F,
\end{equation}
where the weights $w_{ij}, w_{ik} \in  \mathbb{C}$ are designed as
\begin{equation}\label{weight}
    w_{ij} = \frac{(r_j-r_i)^H}{d_{ij}^2}, \ w_{ik} = -\frac{(r_k-r_i)^H}{d_{ik}^2}.
\end{equation}

It yields from \eqref{complexc} that
\begin{equation}\label{complexc1}
    (w_{ij} + w_{ik}) r_i - w_{ij} r_j - w_{ik}r_k = 0.
\end{equation}

We can aggregate the complex constraint of each nominal follower from  \eqref{complexc1} into a matrix form, i.e.,
\begin{equation}\label{g1}
    W_f r = \mathbf{0},
\end{equation}
where $W_f \! \in \! \mathbb{C}^{(n\!-\!m)\times n}$, $r=(r_1^H, \cdots, r_n^H)^H$, and
\begin{equation}\label{zer1}
\begin{array}{ll}
    &[W_f]_{ij} \!=\!  \left\{ \! \begin{array}{lll} 
        \sum\limits_{k \in \mathcal{N}_i} w_{ik},
    & i \in \mathcal{V}_F, \ j=i, \\
    \ \  0, &   i \in \mathcal{V}_F, \ j \notin \mathcal{N}_i, \ j \neq i, \\

    -w_{ij}, & i \in \mathcal{V}_F, \
     j \in \mathcal{N}_i, \ j \neq i. \\
    \end{array}\right. 
\end{array} 
\end{equation}

Based on the leader group $r_L$ and  follower group $r_F$ in $(\mathcal{G}, r)$,  \eqref{g1} becomes
\begin{equation}\label{loi}
   W_{fl} r_L +    W_{f\!f} r_F= \mathbf{0},
\end{equation}
where  $ W_{fl} \in \mathbb{C}^{(n\!-\!m)\times m}$, $W_{f\!f} \in \mathbb{C}^{(n\!-\!m)\times (n\!-\!m)}$ are complex matrices with $ W_f \!=\![
W_{fl} \  W_{f\!f} 
] $. If the complex matrix $W_{f\!f}$ is invertible, the follower group $r_F$ is uniquely determined by the leader group $r_L$, i.e., $r_F = - W_{f\!f}^{-1}W_{fl}r_L$.

\begin{defn}\label{2dd}
A nominal formation $(\mathcal{G}, r)$ is called localizable if the follower group $r_F$ is uniquely determined by the leader group $r_L$, i.e., the matrix $W_{f\!f}$ in \eqref{loi} is invertible.
\end{defn}

\subsection{Target Formation and Control Objective}

Denote $s(t) \!=\! [s_L^H(t), s_F^H(t)]^H \! \in \! \mathbb{C}^n$ as the formation shape parameter, which is designed based on the nominal configuration $r$, i.e., 
\begin{equation}\label{vcon}
W_{fl} s_L(t) +    W_{f\!f} s_F(t) = \mathbf{0},
\end{equation}
where the complex matrices $W_{fl}$ and $W_{f\!f}$ are calculated based on $(\mathcal{G}, r)$ given in \eqref{loi}. The translation parameter, scaling parameter, and 
rotation parameter of the formation are denoted by $\beta(t) \! \in \! \mathbb{C}$, $h(t) \! \in \! \mathbb{R}$, and $\theta(t) \! \in \! \mathbb{R}$, respectively.
Then, the time-varying target formation $(\mathcal{G}, p^*(t))$  is designed as
\begin{equation}\label{ti}
    p^*(t)= {\mathbf{1}}_n \otimes \beta(t) + h(t)[I_n \otimes \text{exp}(\theta(t)\iota)]s(t).
\end{equation}

It is clear in \eqref{ti} that 
$p^*(t)$ and $s(t)$ have the same formation shape. Given any $s_L(t)$, we can find a solution $s_F(t)$ to \eqref{vcon} if the matrix $W_{f\!f}$ is invertible. Thus,
the design of $s_L(t)$ determines
the formation shape $s(t)$. The permitted formation shape $s(t) $ belongs to the null space of the matrix $W_f$, i.e., $W_fs(t)=\mathbf{0}$. The nominal formation $r$ and target formation $p^*(t)$
 have the same formation shape if the formation shape parameter $s(t)$ in \eqref{ti} is set as $s(t)=r$.

\begin{lemma}\label{le1}
(\textbf{Property of Complex Constraint}) 
The time-varying target formation $(\mathcal{G}, p^*(t))$ in \eqref{ti} is localizable if its constant nominal formation $(\mathcal{G}, r)$ is localizable. 
\end{lemma}

\begin{proof}

We can know from \eqref{zer1} that each row sum of $W_f$ is zero, i.e., $W_f \cdot \mathbf{1}_{n} = \mathbf{0}$. Since $ W_f \!=\![
W_{fl} \ \ W_{f\!f}] $, we obtain
\begin{equation}\label{tto}
 W_f \cdot \mathbf{1}_{n}  \!=\![W_{fl} \ \ W_{f\!f} ] \cdot \mathbf{1}_{n} \!=\!  W_{fl} \cdot \mathbf{1}_{m} \!+\! W_{f\!f} \cdot \mathbf{1}_{n\!-\!m} \!=\! \mathbf{0}.
\end{equation}

It yields from \eqref{vcon}, \eqref{ti}, and \eqref{tto} that
\begin{equation}\label{ty1}
    \begin{array}{ll}
         &  \hspace{-0.3cm} W_{f\!f}p^*_F(t) + W_{fl} p^*_L(t) \\
         & \hspace{-0.3cm} =    [W_{fl} \cdot \mathbf{1}_{m} + W_{f\!f} \cdot \mathbf{1}_{n\!-\!m} ]\otimes \beta(t) \\
         & \hspace{-0.3cm} \ \ \  + h(t) [I_{n\!-\!m} \otimes \text{exp}(\theta(t)\iota)]
         [W_{fl} s_L(t) \!+\! W_{f\!f}s_F(t)] \\
         & \hspace{-0.3cm} = \mathbf{0}.
    \end{array}
\end{equation}

If  $(\mathcal{G}, r)$ is localizable, i.e., the complex matrix $ W_{f\!f}$ is invertible, we have $p_F^*(t) = -W_{f\!f}^{-1}W_{fl} p_L^*(t)$.
Thus, the target formation $(\mathcal{G}, p^*(t))$ in \eqref{ti} is localizable.

\end{proof}

It yields from \eqref{ti} and \eqref{ty1} that the target positions of the leader group and follower group are
\begin{equation}\label{pro}
\begin{array}{ll}
     &  p_L^*(t) = {\mathbf{1}}_m \otimes \beta(t) + h(t)[I_n \otimes \text{exp}(\theta(t)\iota)]s_L(t) ,  \\
     &  p_F^*(t) = -W_{f\!f}^{-1}W_{fl} p_L^*(t). 
\end{array}
\end{equation}

Then, the control objectives of the agents are
\begin{align} 
& \label{obj1}   \lim\limits_{t \rightarrow \infty}(p_L(t) - p_L^*(t))= \mathbf{0}. \\
& \label{obj2} \lim\limits_{t \rightarrow \infty}(p_F(t) + W_{f\!f}^{-1}W_{fl} p_L(t))=\mathbf{0}.  
\end{align}

\section{Formation Maneuver Control in 2-D Plane}\label{2dp}

Each agent $i$ in a formation $(\mathcal{G}, p)$ of $n$ agents is described by
\begin{equation}\label{sys}
        \dot p_i = v_i, \ \ i =1, \cdots, n,
    \end{equation}
where $v_i \in \mathbb{C}$ is the control input. 

\subsection{Control Protocol of the Leaders}

\begin{assumption}\label{lle1}
 Only the leaders have access to their positions and target positions, and the target velocity of each leader $i$ is bounded, i.e., there exists $\delta \! \in \! \mathbb{R}$ such that $\| \dot p_i^*(t) \|_2 \le \delta$.
\end{assumption}

Only the leader group knows the parameters 
$\beta(t),h(t),\theta(t),s(t)$ in \eqref{ti}. It yields from \eqref{obj1} that each leader $i \! \in \! \mathcal{V}$ aims at reaching its target position $p_i^* \!=\! x_i^* \!+\! y_i^* \iota$, i.e., $\lim\limits_{t \rightarrow \infty}\! p_i(t) \!=\! p^*_i(t), i \in \mathcal{V}_L$. To achieve this goal, the control protocol of each leader is designed as
\begin{equation}\label{vi2}
    v_i \!=\! -\tanh(x_i \!-\! x_i^*) - \iota \cdot \tanh(y_i \!-\! y_i^*)+ \dot p^*_i, \ i \in \mathcal{V}_L,
\end{equation}
where $\dot p_i^* \!=\! \dot x_i^* \!+\!  \dot y_i^* \iota$. $\tanh(\cdot)$ is the hyperbolic tangent function in $\mathbb{R}$. \eqref{vi2} is equivalent to
\begin{equation}\label{vc2}
 \left\{ \! \! \begin{array}{lll} 
     \dot x_{i} \!=\! - \tanh(x_i \!-\! x_i^*) + \dot x_i^*,  \\
     \dot y_{i} \!=\! - \tanh(y_i \!-\! y_i^*) + \dot y_i^*.
    \end{array}\right.
\end{equation}

\begin{lemma}\label{2l}
Under Assumption \ref{lle1}, the leader group achieves their control objective \eqref{obj1}  under control protocol \eqref{vi2}.
\end{lemma}

\begin{proof}

Denote by $e_{i}\!=\!p_i\!-\!p^*_i$
the tracking error of leader $i$.  Consider a Lyapunov function 
\begin{equation}\label{vl1}
    V_1 = \sum\limits_{i=1}^m \frac{1}{2}[(x_i \!-\! x_i^*)^2+(y_i \!-\! y_i^*)^2].
\end{equation}

We have
\begin{equation}
\begin{array}{ll}\label{lea}
    \hspace*{-0.2cm} \dot V_1\hspace*{-0.3cm} & \!=\!  - \sum\limits_{i=1}^m  
     [(x_i \!-\! x_i^*) \tanh(x_i \!-\! x_i^*) \!+\! (y_i \!-\! y_i^*)\tanh(y_i \!-\! y_i^*)].
\end{array}
\end{equation}

Note that $(x_i \!-\! x_i^*) \tanh(x_i \!-\! x_i^*), (y_i \!-\! y_i^*)\tanh(y_i \!-\! y_i^*) > 0$ if $x_i \!-\! x_i^*, y_i \!-\! y_i^* \neq 0$. Thus, $\dot V_1 < 0$ if $e_i\neq 0, i \in \mathcal{V}_L$. Then, the conclusion follows. 
\end{proof}

\subsection{Control Protocol of the Followers}

It is shown in \eqref{ty1} that the complex matrices $W_{f\!f}$ and $W_{fl}$ are invariant to translation $\beta(t)$,  scaling $h(t)$, rotation $\theta(t)$, and shape change $s(t)$ of the nominal formation $(\mathcal{G}, r)$. That is, the followers can use constant complex weights in $W_{f\!f}$ and $W_{fl}$ to realize formation maneuver control without the need of knowing the parameters $\beta(t),h(t),\theta(t),s(t)$ in \eqref{ti}.

\begin{assumption}\label{lle2}
The constant nominal formation $(\mathcal{G}, r)$ is localizable. In addition, for each nominal follower $i$ and its two neighbors $j,k \! \in \! \mathcal{N}_i$, it has $r_i \neq  r_j \neq r_k$. 
\end{assumption}

From the work \cite{lin2016distributedad}, to guarantee a localizable nominal formation $(\mathcal{G}, r)$ with complex constraints,  the graph condition is that each nominal follower is two-reachable from the nominal leader group. It is concluded from \eqref{weight} that the complex weights satisfy $w_{ij}+w_{ik} \neq 0$ if $r_j \neq r_k$. Then, the control protocol of each follower $i \in \mathcal{V}_F$ is designed as
\begin{equation}\label{ff1}
    v_i \!=\! \frac{w_{ij}(v_j - \alpha_1(p_i \!-\!p_j))}{w_{ij}+w_{ik}} + \frac{w_{ik}(v_k - \alpha_1(p_i\!-\!p_k))}{w_{ij}+w_{ik}},
\end{equation}
where $j,k \! \in \! \mathcal{N}_i$. $\alpha_1 \in \mathbb{R}$ is a positive real control gain.  \eqref{ff1} can be expressed in a matrix form.
\begin{equation}\label{ff2}
    W_{f\!f} \dot p_F + W_{fl} \dot p_L = - \alpha_1( W_{f\!f}p_F +  W_{fl}p_L).
\end{equation}

\begin{theorem}\label{3ll}
Under Assumption \ref{lle2}, the follower group achieves their control objective \eqref{obj2}  under the control protocol \eqref{ff1}.
\end{theorem}

\begin{proof}
Under Assumption \ref{lle2}, the complex matrix $W_{f\!f}$ is invertible.
Based on \eqref{obj2}, the tracking error of the follower group is given by
\begin{equation}\label{eff1}
  e_F \!=\! p_F + W_{f\!f}^{-1}W_{fl} p_L = e_{f_x} + e_{f_y}\iota,
\end{equation}
where $e_{f_x},e_{f_y} \in \mathbb{R}^{n\!-\!m}$ represent the real and imaginary part of $e_F$, respectively. Then, \eqref{ff2} is equivalent to
\begin{equation}\label{ef1}
  \dot e_{f_x} \!=\! -\alpha_1 e_{f_x}, \  \ \dot e_{f_y} \!=\! -\alpha_1 e_{f_y}.
\end{equation}

Consider a Lyapunov function 
\begin{equation}\label{vk2}
    V_2 = e_{f_x}^T e_{f_x} + e_{f_y}^Te_{f_y}.
\end{equation}

It yields from \eqref{ef1} that
\begin{equation}
\dot V_2 \!=\! e_{f_x}^T \dot e_{f_x} \!+\! e_{f_y}^T \dot e_{f_y} \!=\! - 2\alpha_1 (e_{f_x}^2 \!+\! e_{f_y}^2) < 0,  \ \text{if} \ e_{f_x}, e_{f_y} \neq 0.   
\end{equation}

Hence, the tracking error $e_F$ of the follower group converges to zero exponentially fast. 
\end{proof}

The controller of each follower in \eqref{ff1} requires velocity and relative position information. Next, we will introduce how the followers can achieve their control objectives by only using relative position information. Define a nonlinear function $\digamma \! (\cdot)$ as: for any complex number $p_i \in \mathbb{C}$,
\begin{equation}\label{sig}
\begin{array}{ll}
    &    \digamma \! (p_i) \!=\!  \left\{ \! \begin{array}{lll} 
    \frac{p_i}{\| p_i \|_2}, & 
    p_i  \neq 0, \\ 0, & 
    p_i =0.  
    \end{array}\right. 
\end{array} 
\end{equation} 

Then, the control protocol of each follower $i \! \in \! \mathcal{V}_F$ is designed as
\begin{equation}\label{ff3}
v_i \!=\! \gamma_{{ijk}} + \hspace{-0.3cm} \sum\limits_{i,l \in \mathcal{N}_g, g \in \mathcal
     {V}_f} \hspace{-0.5cm} \gamma_{{gil}} + \alpha_2 \digamma \! \! \left(\gamma_{{ijk}} + \hspace{-0.3cm} \sum\limits_{i,l \in \mathcal{N}_g, g \in \mathcal
     {V}_f} \hspace{-0.5cm} \gamma_{{gil}}\right),
\end{equation}
where $(g, i), (g,l) \in  \mathcal{E}$. $\alpha_2 \in \mathbb{R}$ is the control gain  and
\begin{equation}\label{disj}
    \begin{array}{ll}
         &  \gamma_{ijk} \!=\! (w_{ij}\!+\!w_{ik})^H[w_{ij}(p_j \!-\!  p_i) \!+\! w_{ik}( p_k \!-\!  p_i)], \ j,k \in \mathcal{N}_i, \\
         & \gamma_{gil} = w_{gi}^Hw_{gi}(p_g \!-\!  p_i) \!+\! w_{gi}^Hw_{gs}(p_g \!-\! p_l).
    \end{array}
\end{equation}

To implement distributed localization protocol \eqref{ff3}, we need to make the directed edges among the followers in $(\mathcal{G}, r)$ to be undirected.

\begin{remark}
For follower $i$, the information $\gamma_{ijk}$ is obtained based on the relative positions between neighboring agents through the directed edges $(i,j),(i,k) \! \in \! \mathcal{E}$. 
If follower $i$ is a neighbor of follower $g$ in $\mathcal{G}$, i.e., $(g,i) \! \in \! \mathcal{E}$, follower $i$ needs to obtain
the information $\gamma_{gjl}$ from follower $g$ by communication through the edge $(i,g) \! \in \! \mathcal{E}$, i.e., the edges among the followers need to be undirected.
\end{remark}

Let 
\begin{equation}
\zeta_F=D_{f\!f}p_F \!+\! D_{fl}p_L,   
\end{equation}
where $\zeta_F=[\zeta_1^H, \cdots, \zeta^H_{n\!-\!m}]^{H} \in \mathbb{C}^{n\!-\!m}$ and
\begin{equation}\label{dsl1}
 D_{f\!f} =  W_{f\!f}^H W_{f\!f}, \   D_{fl} =  W_{f\!f}^H W_{fl}.  
\end{equation}

If  the matrix $W_{f\!f}$ is invertible, we have  $\zeta_F=D_{f\!f}e_F$. Note that \eqref{ff3} can be expressed in a matrix form.
\begin{equation}\label{ff4}
\dot p_F = -D_{f\!f}p_F-D_{fl}p_L - \alpha_2 \zeta,
\end{equation}
where $\zeta = [  [\digamma \!  (\zeta_{1})]^H, \  \cdots, \ [\digamma \! (\zeta_{n\!-\!m})]^H  ]^H$.

\begin{theorem}\label{2df}
Suppose Assumptions \ref{lle1}-\ref{lle2} hold. The follower group under control protocol \eqref{ff3} can achieve their control objective \eqref{obj2} if the subgraph among the followers is undirected and
\begin{equation}\label{fgh3}
    \alpha_2 \ge \xi (\sqrt{2} + \delta)m,
\end{equation}
where $\xi \!=\! \max\limits_{i,j}\| \xi_{ij} \|_2$ and $\xi_{ij}$ is the $(i,j)$-th entry of the complex matrix $W_{f\!f}^{-1} W_{fl}$.
\end{theorem}

\begin{proof}
Under Assumption \ref{lle2}, the complex matrix $W_{f\!f}$ is invertible, i.e., the complex matrix $D_{f\!f} \!>\! 0$ is positive definite. 
Combining \eqref{eff1} and \eqref{ff4}, we have
\begin{equation}\label{ff5}
    \dot e_F = \dot p_F  + W^{-1}_{f\!f}W_{fl} \dot p_L
=  -D_{f\!f}e_F  - \alpha_2 \zeta + W^{-1}_{f\!f}W_{fl}v_L,
\end{equation}
where $v_L \!=\! [v_1^H, \cdots, v_m^H]^H$ are the control inputs of the leaders \eqref{vi2}. Let the real and imaginary parts of $D_{f\!f}, \zeta,$ $W^{-1}_{f\!f}W_{fl}v_L$ be
\begin{equation}\label{eq}
\begin{array}{ll}
     & D_{f\!f} = D_{x} + D_{y}\iota, \ \zeta =\zeta_x +\zeta_y \iota, \\
     & W^{-1}_{f\!f}W_{fl}v_L = v_{l_x} + v_{l_y}\iota = [v_{l_1}^H, \cdots, v_{l_{n\!-\!m}}^H ]^H,
\end{array}
\end{equation}
where $D_{x},D_{y} \!\in \! \mathbb{R}^{(n\!-\!m)\times(n\!-\!m)}$ and $v_{l_x},v_{l_y},\zeta_x,\zeta_y \!\in\! \mathbb{R}^{n\!-\!m}$. $v_{l_i}\! \! \in \!  \mathbb{C}(i\!=\!1,\cdots,n\!-\!m) $. Since $D_{f\!f}^H \!=\! D_{f\!f}$, we have $D^T_{x} \!=\! D_{x}$ and $D^T_{y} \!=\!-D_{y}$. Under Assumption \ref{lle1} and from \eqref{vi2}, we obtain
\begin{equation}\label{bd}
    \| v_i \|_2 \le \sqrt{2} + \delta, \ i = 1, \cdots, m.
\end{equation}

Observe from \eqref{bd} that the velocities of the leaders are bounded by $\sqrt{2} + \delta$. Hence, \eqref{fgh3} basically indicates that the control gain $a_2$ in \eqref{ff3} should be larger than a constant which is related to the velocity bound $\sqrt{2} \!+\! \delta$.
It yields from \eqref{eq} and \eqref{bd} that 
\begin{equation}\label{eq1}
    \|v_{l_i} \|_2 \le \sum\limits_{j=1}^{m}\xi \|v_j\|_2 \le \xi (\sqrt{2} + \delta)m, \ i = 1, \cdots, n\!-\!m.
\end{equation}

Based on \eqref{eff1} and \eqref{eq},  \eqref{ff5} is equivalent to
\begin{equation}\label{e2}
 \left\{ \! \! \begin{array}{lll} 
     \dot e_{f_x} \!=\! -D_{x} e_{f_x} \!+\! D_{y}e_{f_y} \!-\! \alpha_2\zeta_x \!+\! v_{l_x},  \\
     \dot e_{f_y} \!=\! -D_{x} e_{f_y} \!-\! D_{y}e_{f_x} \!-\! \alpha_2\zeta_y \!+\! v_{l_y}.
    \end{array}\right.
\end{equation}

Consider a Lyapunov function 
\begin{equation}\label{vk3}
    V_3 = \frac{1}{2}e_{f_x}^T D_x e_{f_x} + e_{f_y}^T D_y e_{f_x} + \frac{1}{2}e_{f_y}^TD_xe_{f_y}.
\end{equation}

Since $e_{f_x},e_{f_y},D_x,D_y$ are real vectors or matrices, the function $V_3 \in \mathbb{R}$ belongs to real space and 
$V_3 = \frac{1}{2}e_F^H D_{f\!f} e_F \ge 0$. The function $\digamma \! (\cdot)$ \eqref{sig} is discontinuous. Based on \eqref{eq1} and \eqref{e2}, the set-valued Lie derivative of $V_3$ is
\begin{equation}\label{v3}
\begin{array}{ll}
\dot{\tilde{V}}_3  \hspace{-0.3cm} &= \mathcal{K}[e_{f_x}^T D_x \dot e_{f_x} \!+\! e_{f_y}^T D_y \dot e_{f_x} \!+\! \dot e_{f_y}^T D_y e_{f_x} \!+\! e_{f_y}^T D_x \dot e_{f_y}] \\
& = \mathcal{K}[-e_F^HD^2_{f\!f}e_F\! -\! (D_xe_{f_x}\!-\!D_ye_{f_y})^T(\alpha_2\zeta_x-v_{l_x}) \\
& \ \ \  -(D_xe_{f_y}\!+\!D_ye_{f_x})^T(\alpha_2\zeta_y- v_{l_y})].
\end{array}
\end{equation}

Note that $(D_xe_{f_x}\!-\!D_ye_{f_y})^T\zeta_x \!+\! (D_xe_{f_y}\!+\!D_ye_{f_x})^T\zeta_y \!=\! \sum\limits_{j=1}^{n\!-\!m} \| \zeta_j\|_2  $ is continuous. Since  $\mathcal{K}[\varpi]=\{ \varpi\}$ if $\varpi$ is continuous, it yields from \eqref{v3} that
\begin{equation}\label{vp3}
\begin{array}{ll}
\dot{\tilde{V}}_3  \hspace{-0.3cm} & = -e_F^HD^2_{f\!f}e_F\! -\! (D_xe_{f_x}\!-\!D_ye_{f_y})^T(\alpha_2\zeta_x-v_{l_x})\\
& \ \ \  -(D_xe_{f_y}\!+\!D_ye_{f_x})^T(\alpha_2\zeta_y- v_{l_y})   \\
& \le -e_F^HD^2_{f\!f}e_F -\alpha_2\sum\limits_{j=1}^{n\!-\!m} \| \zeta_j\|_2 + \sum\limits_{j=1}^{n\!-\!m} \| \zeta_j\|_2 \|v_{l_j}\|_2    \\
& \le -e_F^HD^2_{f\!f}e_F -(\alpha_2-\xi (\sqrt{2}+ \delta)m)\sum\limits_{j=1}^{n\!-\!m} \| \zeta_j\|_2     \\
& \le -e_F^HD^2_{f\!f}e_F < 0, \ \text{if} \ e_{f_x}, e_{f_y} \neq 0.
\end{array}
\end{equation}

From Theorem $3.2$ in \cite{shevitz1994lyapunov}, we can know from \eqref{e2} and \eqref{vp3}  that 
$\lim\limits_{t\rightarrow \infty}e_F(t)=0$. 
\end{proof}

Next, we will investigate how to achieve inter-agent collision avoidance.

\begin{theorem}\label{intas}
(\textbf{Inter-agent Collision Avoidance}) 
Suppose Assumptions \ref{lle1}-\ref{lle2} hold.  The inter-agent collision is avoided and the control objectives \eqref{obj1} and \eqref{obj2} are achieved  under the controller of the leaders \eqref{vi2} and the controller of the followers \eqref{ff1} or \eqref{ff3} if the initial positions satisfy $\|p_i(0) \!-\! p_j(0)\|_2 \neq 0$ and the target positions satisfy $\|p_i^*(t) \!-\! p_j^*(t)\|_2 \!-\!  \psi_i \!-\! \psi_j > 0$, $ \text{for any} \ i,j \in \mathcal{V}$, where \begin{equation}\label{kh}
\begin{array}{ll}
    &\psi_i  \!=\!  \left\{ \! \begin{array}{lll} 
    \|p_i(0) - p_i^*(0)\|_2, & i \in \mathcal{V}_L,  \\
  \sqrt{\frac{\lambda_{\max}(D_{f\!f})}{\lambda_{\min}(D_{f\!f})}}\| e_F(0) \|_2,   &   i \in \mathcal{V}_F. \\
    \end{array}\right. 
\end{array} 
\end{equation}
\end{theorem}

\begin{proof}

For any $i,j \in \mathcal{V}$ and $t\ge0$, we have
\begin{equation}\label{eni}
    p_i- p_j  =  (p_i - p_i^*) + (p_j^*-p_j) + (p_i^*-p_j^*). 
\end{equation}

From  \eqref{vl1}, \eqref{vk2}, and \eqref{vk3}, we have $\|p_i(t)-p_i^*(t) \|  \! \le  
      \! \psi_i, \ i \in \mathcal{V}$.
Then, \eqref{eni} becomes
\begin{equation}\label{trw}
\begin{array}{ll}
     &   \| p_i- p_j \|_2 \ge \|p_i^*-p_j^*\|_2 -  \|p_i - p_i^*\|_2 - \|p_j-p^*_j\|_2 \\
     & \ \ \ \ \ \ \ \ \ \ \ \ \  \ge \|p_i^*-p_j^*\|_2 \!-\!  \psi_i \!-\! \psi_j > 0, t \ge 0.
\end{array}
\end{equation}

Therefore, the inter-agent collision is avoided. From Lemma \ref{2l}, Theorem \ref{3ll}, and Theorem \ref{2df}, we can know that the agents can achieve the control objectives \eqref{obj1} and \eqref{obj2}. 

\end{proof}

\section{Formation Maneuver Control in 3-D Space}\label{3dp}

The 3-D position $(x_i,y_i,z_i)^T \! \in \! \mathbb{R}^3$ of agent $i$ can be represented by a 2-D complex vector, i.e.,
\begin{equation}\label{comp3}
    \rho_i = \left[\begin{array}{ll}
   p_i \\
   \tau_i
    \end{array}\right] = \left[\begin{array}{ll}
  x_i + y_i\iota \\
   \ \ \ z_i
    \end{array}\right], \ i=1, \cdots, n.
\end{equation}

Each agent $i$ in 3-D space is then described by
\begin{equation}\label{3dm}
        \dot {\rho}_i = \mu_i, \ \ i =1, \cdots, n,
    \end{equation}
where $\mu_i \!=\! [v_{i}^H, \sigma_i]^H \in \mathbb{C}^2$ is the control input with $v_i \in \mathbb{C}$ and $\sigma_i \in \mathbb{R}$.  The agent dynamics \eqref{3dm} can be rewritten as two subsystems, i.e.,
\begin{equation}\label{3dm12}
\dot p_i = v_i, \  \ \dot \tau_i = \sigma_i.
\end{equation}

Denote $r \!=\! [r_1^H, \cdots, r_{n}^H]^H \in \mathbb{C}^n$ and $ \epsilon \!=\! [\epsilon _1, \cdots, \epsilon _{n}]^T \in \mathbb{R}^n$ as the nominal positions of the agents in the X-Y plane and Z-axis, respectively. For a constant nominal formation $(\mathcal{G}, [r, \epsilon])$, we can calculate its weight matrices like \eqref{loi}, i.e.,
\begin{equation}\label{loi1}
   W_{fl} r_L +    W_{f\!f} r_F = \mathbf{0}, \    M_{fl} \epsilon_L +    M_{f\!f} \epsilon_F = \mathbf{0},
\end{equation}
where $ W_{fl} \! \in \! \mathbb{C}^{(n\!-\!m)\times m}$, $W_{f\!f} \! \in \! \mathbb{C}^{(n\!-\!m)\times (n\!-\!m)}$ are complex matrices,  and $ M_{fl} \! \in \! \mathbb{R}^{(n\!-\!m)\times m}$, $M_{f\!f} \! \in \! \mathbb{R}^{(n\!-\!m)\times (n\!-\!m)}$ are real matrices. $r_L \!=\! [r_1^H, \cdots, r_{m}^H]^H$ and 
$r_F \!=\! [r_{m\!+\!1}^H, \cdots,r_{n}^H]^H$ are the nominal positions of the leader group and follower group in X-Y plane, respectively. 
$\epsilon_L \!=\! [\epsilon_1, \cdots, \epsilon _{m}]^T$ 
and $\epsilon_F \! = \! [\epsilon_{m\!+\!1}, \cdots,  \epsilon _{n}]^T$ are
the nominal positions of the leader group and follower group in Z-axis, respectively. Let $p^*(t) \!=\! [[p^*_1(t)]^H, \cdots, [p^*_{n}(t)]^H]^H \! \in \! \mathbb{C}^n$ and $\tau^*(t) \!=\! [\tau^*_1(t), \cdots, \tau^*_{n}(t)]^T $ $\! \in \! \mathbb{R}^n$ be the time-varying target positions of the agents in the X-Y plane and Z-axis, respectively. Then,
the time-varying target formation $(\mathcal{G}, [p^*(t), \tau^*(t)])$ in 3-D space is given by
\begin{equation}\label{ti1}
 \left\{ \! \! \begin{array}{lll}
    p^*(t)= {\mathbf{1}}_n \otimes \beta_{p}(t) + h(t)[I_n \otimes \text{exp}(\theta_{p}(t)\iota)]s_{p}(t), \\
     \tau^*(t) = {\mathbf{1}}_n \otimes \beta_{\tau}(t) + h(t)s_{\tau}(t), 
    \end{array}\right.
\end{equation}
where  $\beta_{p}(t) \! \in \! \mathbb{C}, \beta_{\tau}(t) \! \in \! \mathbb{R}$ are translation parameters in X-Y plane and Z-axis, respectively. $h(t) \! \in \! \mathbb{R}$ is scaling parameter, and $\theta_{p}(t) \! \in \! \mathbb{R}$ is rotation parameter in X-Y plane. $s_{p}(t) \!=\! [[s^{p}_{L}(t)]^H, [s^{p}_{F}(t)]^H]^H \! \in \! \mathbb{C}^n$ and $s_{\tau}(t) \!=\! [[s^{\tau}_{L}(t)]^H, [s^{\tau}_{F}(t)]^H]^H \in \mathbb{C}^n$ are, respectively, the formation shape parameters in X-Y plane and Z-axis, where
\begin{equation}\label{vcon1}
W_{fl} s^{p}_{L}(t)  \!+ \!   W_{f\!f} s^{p}_{F}(t) \!=\! \mathbf{0}, \ M_{fl} s^{\tau}_{L}(t) \!+ \!   M_{f\!f} s^{\tau}_{F}(t) \!=\! \mathbf{0}.
\end{equation}

\begin{defn}
A 3-D nominal formation $(\mathcal{G}, [r, \epsilon])$ is called localizable if the matrices $W_{f\!f}$ and $M_{f\!f}$ in \eqref{loi1} are invertible.
\end{defn}

If the constant nominal formation $(\mathcal{G}, [r, \epsilon])$ is localizable, 
we can conclude from \eqref{ti1} and Lemma \ref{le1} that the time-varying target formation $(\mathcal{G}, [p^*(t), \tau^*(t)])$ is also localizable, i.e.,
\begin{equation}\label{re1}
p_F^*(t) \!=\! -W_{f\!f}^{-1}W_{fl} p_L^*(t), \ \tau_F^*(t) \!=\! -M_{f\!f}^{-1}M_{fl} \tau_L^*(t).
\end{equation}

Then, the control objective of the agents in 3-D space is
\begin{align}
\label{obj11} \lim\limits_{t \rightarrow \infty}(p_L(t) - p_L^*(t))= \mathbf{0}, \  \lim\limits_{t \rightarrow \infty}(\tau_L(t) - \tau_L^*(t))= \mathbf{0}. \\
 \label{obj21}  \left\{ \! \! \begin{array}{lll} 
 \lim\limits_{t \rightarrow \infty}(p_F(t) + W_{f\!f}^{-1}W_{fl} p_L(t))=\mathbf{0}, \\
\lim\limits_{t \rightarrow \infty}(\tau_F(t) + M_{f\!f}^{-1}M_{fl} \tau_L(t))=\mathbf{0}. 
    \end{array} \right.
\end{align}

\subsection{Control Protocol}\label{pro3}

\begin{assumption}\label{lle13}
Only the leaders have access to their positions and target positions, and the target velocity of each leader $i$ is bounded, i.e., there exists $\delta_{\mu} \! \in \! \mathbb{R}$ such that $\| \dot p_i^*(t) \|_2, \| \dot \tau_i^*(t) \|_2 \! \le \! \delta_{\mu}$.
\end{assumption}

\begin{assumption}\label{lle23}
The constant nominal formation $(\mathcal{G}, [r, \epsilon])$ is localizable. In addition, for each nominal follower $i$ and its two neighbors $j,k \! \in \! \mathcal{N}_i$, it has $r_i \neq r_j \neq r_k$ and $\epsilon_i \neq \epsilon_j \neq \epsilon_k $.
\end{assumption}

Note that the agent dynamics is decomposed into two 2-D subsystems \eqref{3dm12}. The two-reachable graphs can also be used to guarantee a 3-D  nominal formation $(\mathcal{G}, [r, \epsilon])$ to be localizable \cite{lin2016distributedad}. Each leader $i \in \mathcal{V}_L$ has access to its target position $\rho_i^* \!=\! [[p^*_i]^H, \tau^*_i]^H$. The control protocol of each leader $i \in \mathcal{V}_L$ in \eqref{3dm12} is designed as
\begin{equation}\label{vi23}
\hspace{-0.15cm} \left\{ \! \! \begin{array}{lll} 
      v_i \!=\! -\tanh[ \text{real}(p_i \!-\! p_i^*)] \!-\! \iota \cdot \tanh[\text{imag}(p_i \!-\! p_i^*)] \!+\! \dot p^*_i,     \\
      \sigma_i \!= \! -\tanh(\tau_i \!-\! \tau_i^*) + \dot \tau^*_i,
    \end{array} \right.
\end{equation}
where $\text{real}(\cdot)$ and $\text{imag}(\cdot)$ represent the real-part and imaginary part of complex number, respectively.

\begin{lemma}\label{3l}
Under Assumption \ref{lle13}, the leader group achieves their control objective \eqref{obj11}  under control protocol \eqref{vi23}.
\end{lemma}

The proof of Lemma \ref{3l} follows from that of Lemma \ref{2l}. Denote $W_f \!=\! [W_{fl} \ W_{f\!f}]$ as the complex matrix in X-Y plane. Let  $M_f \!=\! [M_{fl} \ M_{f\!f}]$ be the real matrix in Z-axis. Let $w^{p}_{ij}$ and $w^{\tau}_{ij}$ be the $(i,j)$th element of matrices $W_f$ and $M_f$, respectively. Then, the control protocol of each follower $i \in \mathcal{V}_F$ in \eqref{3dm12} is designed as
\begin{equation}\label{vif23}
 \hspace{-0.2cm} \left\{ \! \! \begin{array}{lll}
v_i \!=\! \gamma_{{ijk}} + \hspace{-0.3cm} \sum\limits_{i,l \in \mathcal{N}_g, g \in \mathcal
     {V}_f} \hspace{-0.5cm} \gamma_{{gil}} + \alpha_2 \digamma \! \! \left(\gamma_{{ijk}} + \hspace{-0.3cm} \sum\limits_{i,l \in \mathcal{N}_g, g \in \mathcal
     {V}_f} \hspace{-0.5cm} \gamma_{{gil}}\right),  \\
\sigma_i =  \eta_{{ijk}} + \hspace{-0.3cm} \sum\limits_{i,l \in \mathcal{N}_g, g \in \mathcal
     {V}_f} \hspace{-0.5cm} \eta_{{gil}} + \alpha_2 \digamma \! \! \left(\eta_{{ijk}} + \hspace{-0.3cm} \sum\limits_{i,l \in \mathcal{N}_g, g \in \mathcal
     {V}_f} \hspace{-0.5cm} \eta_{{gil}}\right),
    \end{array} \right.
\end{equation}
where $(g, i), (g,l) \in  \mathcal{E}$, $j,k \in \mathcal{N}_i$, and 
\begin{equation}\label{disj1}
 \left\{ \! \! \begin{array}{lll} 
 \gamma_{ijk} \!=\! (w^{p}_{ij}\!+\!w^{p}_{ik})^H[w^{p}_{ij}(p_j \!-\!  p_i) \!+\! w^{p}_{ik}( p_k \!-\!  p_i)], \\
\eta_{ijk} \!=\! (w^{\tau}_{ij}\!+\!w^{\tau}_{ik})[w^{\tau}_{ij}(p_j \!-\!  p_i) \!+\! w^{\tau}_{ik}( p_k \!-\!  p_i)], \\
\gamma_{gil} = [w^{p}_{gi}]^Hw^{p}_{gi}(p_g \!-\!  p_i) \!+\! [w^{p}_{gi}]^Hw^{p}_{gs}(p_g \!-\! p_l), \\
\eta_{gil} = w^{\tau}_{gi}w^{\tau}_{gi}(p_g \!-\!  p_i) \!+\! w^{\tau}_{gi}w^{\tau}_{gs}(p_g \!-\! p_l).
    \end{array} \right.
\end{equation}

\begin{theorem}\label{3df}
Suppose Assumptions \ref{lle13}-\ref{lle23} hold. The follower group  under control protocol \eqref{vif23} can achieve their control objective \eqref{obj21} if the subgraph among the followers is undirected and
\begin{equation}
    \alpha_2 \ge \xi_u (\sqrt{2} + \delta_u)m,
\end{equation}
where $\xi_u = \max\limits_{i,j} \{ \| \xi^{p}_{ij} \|_2, \| \xi^{\tau}_{ij} \|_2 \}$. $\xi^{p}_{ij}$ and $\xi^{\tau}_{ij}$ are the $(i,j)$-th entry of the matrices $W_{f\!f}^{-1} W_{fl}$ and $M_{f\!f}^{-1} M_{fl}$, respectively.
\end{theorem}

The proof of Theorem \ref{3df} follows from that of Theorem \ref{2df}.

\subsection{Any Orientation Change of Formation in 3-D Space}\label{pitch}

The limitation of the proposed method in Section \ref{pro3} is that the orientation of formation can only be changed in X-Y plane, i.e., we can only change the yaw angle of formation. 
The immediate question is how
to rotate the formation with any orientation in 3-D space. Denote $q \!=\! [q_1^T, \cdots, q^T_n]^T \in \mathbb{R}^{3n}$ as a constant nominal configuration in 3-D real space, where $q_i \!=\!(x_{q_i},y_{q_i},z_{q_i})^T \! \in \! \mathbb{R}^3$ is the nominal position of agent $i$. Let
\begin{equation}\label{oliu}
    \left\{ \! \! \begin{array}{lll} 
     r^{a}_i \!=\! x_{q_i} \!+ y_{q_i} \iota, \ \epsilon^{a}_i = z_{q_i}, \\
      r^{b}_i \!=\! x_{q_i} \!+ z_{q_i} \iota, \ \epsilon^{b}_i = y_{q_i}, \\
     r^{c}_i \!=\! y_{q_i} \!+ z_{q_i} \iota, \ \epsilon^{c}_i = x_{q_i}.
    \end{array} \right.
\end{equation}

Based on \eqref{oliu}, we can construct three kinds of constant nominal configuration $(\mathcal{G}, [r^a, \epsilon^a])$, $(\mathcal{G}, [r^b, \epsilon^b])$, $(\mathcal{G}, [r^c, \epsilon^c])$ as
\begin{equation}\label{oliu1}
    \left\{ \! \! \begin{array}{lll} 
  r^a \!=\! [[r^a_1]^H, \cdots, [r^a_{n}]^H]^H, \ \epsilon^a \!=\! [\epsilon^a_1, \cdots, \epsilon^a_{n}]^T, \\
     r^b \!=\! [[r^b_1]^H, \cdots, [r^b_{n}]^H]^H, \ \epsilon^b \!=\! [\epsilon^b_1, \cdots, \epsilon^b_{n}]^T, \\
r^c \!=\! [[r^c_1]^H, \cdots, [r^c_{n}]^H]^H, \ \epsilon^c \!=\! [\epsilon^c_1, \cdots, \epsilon^c_{n}]^T.
    \end{array} \right.
\end{equation}

Similar to \eqref{loi1}, the weight matrices of the nominal formation $(\mathcal{G}, [r^a, \epsilon^a])$, $(\mathcal{G}, [r^b, \epsilon^b])$, $(\mathcal{G}, [r^c, \epsilon^c])$ are designed as
\begin{equation}
    \left\{ \! \! \begin{array}{lll} 
       W^{a}_{fl} r^a_l +    W^{a}_{f\!f} r^a_f = \mathbf{0}, \    M^a_{fl} \epsilon^a_l + M^a_{f\!f} \epsilon^a_f = \mathbf{0}, \\ 
       W^{b}_{fl} r^b_l +    W^{b}_{f\!f} r^b_f = \mathbf{0}, \    M^b_{fl} \epsilon^b_l + M^b_{f\!f} \epsilon^b_f = \mathbf{0}, \\ 
      W^{c}_{fl} r^c_l +    W^{c}_{f\!f} r^c_f = \mathbf{0}, \    M^c_{fl} \epsilon^c_l + M^c_{f\!f} \epsilon^c_f = \mathbf{0}.
    \end{array} \right.
\end{equation}

Let 
\begin{equation}
    \begin{array}{ll}
         & W^a_f= [W^a_{fl} \ W^a_{f\!f}], \   W^b_f= [W^b_{fl} \ W^b_{f\!f}], \  W^c_f= [W^c_{fl} \ W^c_{f\!f}], \\
         & M^a_f= [M^a_{fl} \ M^a_{f\!f}], \ M^b_f= [M^b_{fl} \ M^b_{f\!f}], \ M^c_f= [M^c_{fl} \ M^c_{f\!f}].
    \end{array}
\end{equation}

Note that any orientation of formation can be achieved by changing the yaw angle $\theta_{xy}(t) \! \in \! \mathbb{R}$, pitch angle $\theta_{xz}(t) \! \in \! \mathbb{R}$, and roll angle $\theta_{yz}(t) \! \in \! \mathbb{R}$ of formation sequentially. A switching control system can be used to rotate the formation with any orientation in 3-D space shown as following.

(\romannumeral1) If the leaders want to change the yaw angle of formation within the time interval $(t_1, t_2]$, the rotation parameter in \eqref{ti1} is set as
    $\theta_p(t) \!= \! \theta_{xy}(t)$. $\beta_p(t)$ and $\beta_{\tau}(t)$ in \eqref{ti1} are translation parameters in X-Y plane and Z-axis, respectively. $h(t)$ is the scaling parameter. $s_{p}(t)$ and $s_{\tau}(t)$ are, respectively, the formation shape parameters in X-Y plane and Z-axis, where $W^a_{f} s_{p}(t) \!=\! \mathbf{0}, \ M^a_{f} s_{\tau}(t) \!=\! \mathbf{0}$. The control weights in \eqref{disj1} are set as $w_{ij}^p=[W^a_f]_{ij}$ and $w_{ij}^{\tau}=[M^a_f]_{ij}$ for any $i,j \in \mathcal{V}$.
    
(\romannumeral2) If the leaders want to change the pitch angle of formation within the time interval $(t_1, t_2]$, the rotation parameter in \eqref{ti1} is set as
$\theta_p(t) \!=\! \theta_{xz}(t)$. $\beta_p(t)$ and $\beta_{\tau}(t)$ in \eqref{ti1} are translation parameters in X-Z plane and Y-axis, respectively. $h(t)$ is the scaling parameter. $s_{p}(t)$ and $s_{\tau}(t)$ are, respectively, the formation shape parameters in X-Z plane and Y-axis, where $W^b_{f} s_{p}(t) \!=\! \mathbf{0}, \ M^b_{f} s_{\tau}(t) \!=\! \mathbf{0}$. The control weights in \eqref{disj1} are set as $w_{ij}^p=[W^b_f]_{ij}$ and $w_{ij}^{\tau}=[M^b_f]_{ij}$ for any $i,j \in \mathcal{V}$.

(\romannumeral3) If the leaders want to change the roll angle of formation within the time interval $(t_1, t_2]$, the rotation parameter in \eqref{ti1} is set as
    $\theta_p(t) \!=\! \theta_{yz}(t)$. $\beta_p(t)$ and $\beta_{\tau}(t)$ in \eqref{ti1} are translation parameters in Y-Z plane and X-axis, respectively. $h(t)$ is the scaling parameter. $s_{p}(t)$ and $s_{\tau}(t)$ are, respectively, the formation shape parameters in Y-Z plane and X-axis, where $W^c_{f} s_{p}(t) \!=\! \mathbf{0}, \ M^c_{f} s_{\tau}(t) \!=\! \mathbf{0}$. The control weights in \eqref{disj1} are set as $w_{ij}^p=[W^c_f]_{ij}$ and $w_{ij}^{\tau}=[M^c_f]_{ij}$ for any $i,j \in \mathcal{V}$.

Thus, during the time interval $(t_1, t_2]$, the velocity control of each agent $i$ in $X$-axis, $Y$-axis, and $Z$-axis can be obtained by \eqref{vi23} and \eqref{vif23}, i.e.,
\begin{equation}
 \left\{ \! \! \begin{array}{lll}
    \dot x_i = \text{real}(v_{i}), \   \dot y_i= \text{imag}(v_{i}), \ \dot z_i = \sigma_i, \ \text{if} \  \theta_p(t)= \theta_{xy}(t), \\
  \dot x_i = \text{real}(v_{i}), \   \dot z_i= \text{imag}(v_{i}), \ \dot y_i = \sigma_i, \ \text{if} \  \theta_p(t)= \theta_{xz}(t), \\
  \dot y_i = \text{real}(v_{i}), \   \dot z_i= \text{imag}(v_{i}), \ \dot x_i = \sigma_i, \ \text{if} \  \theta_p(t)= \theta_{yz}(t).
    \end{array}\right.  
\end{equation}

Since the two-reachable graphs can be used to guarantee that the nominal formation $(\mathcal{G}, [r^a, \epsilon^a])$, $(\mathcal{G}, [r^b, \epsilon^b])$, $(\mathcal{G}, [r^c, \epsilon^c])$ to be localizable,  each follower has at least two disjoint paths to the leader set. Hence, if the leaders want to change the yaw, pitch, or roll angle of formation, the followers can be informed to switch their constant control weights in \eqref{disj1} as stated above.

\section{Discussion of the proposed method}\label{dis}

\subsection{Fewer Neighbors and Lesser Communication}

Similar to \eqref{vif23}, the controller \eqref{ff1} can also be extended to 3-D space if the velocity information is available.
Compared with the existing real-Laplacian-based  methods \cite{zhao2018affine, chen2020distributed,  li2020layered, xu2020affine, onuoha2019,han2017fobarycentric,han2015three,fang2021distributed}, the proposed complex-Laplacian-based controllers in \eqref{ff1}, \eqref{ff3}, and \eqref{vif23} require fewer neighbors and lesser communication. 

(\romannumeral1) If the velocity information is available, each follower in \cite{zhao2018affine,li2020layered,xu2020affine,han2017fobarycentric,han2015three,fang2021distributed} needs at least three neighbors in 2-D plane or four neighbors in 3-D space, but each follower under the proposed control protocol \eqref{ff1} only needs two neighbors in both 2-D and 3-D spaces. (\romannumeral2) If the velocity information is unavailable \cite{zhao2018affine,chen2020distributed,onuoha2019}, the communication among the agents is assumed to be undirected. 
Their symmetric real Laplacian matrix $L_f \! \in \! \mathbb{R}^{ n \times n}$ of the agents can also be rewritten as $L_f \!=\!  [L'_{f}]^T  L'_{f}$. Denote  $\mathcal{G}$ and $\mathcal{G}'$ as the graphs of $L_f$ and $L'_{f}$, respectively.
 To guarantee $L_{f}r \!=\!0$, i.e., $L'_f\!=\!0$,  each follower in $\mathcal{G}'$ has at least three neighbors in 2-D plane or four neighbors in 3-D space \cite{zhao2018affine,chen2020distributed,onuoha2019}.
 Since each follower in the proposed $\mathcal{G}'$ only needs two neighbors in both 2-D and 3-D spaces, the proposed control protocol of each follower \eqref{ff3} or \eqref{vif23} over $\mathcal{G}$ requires fewer neighbors than the real-Laplacian-based controllers in \cite{zhao2018affine,chen2020distributed,onuoha2019}.

\subsection{Shape Change and Nominal configuration}\label{shaped}

The time-varying formation shape parameter $s(t)$ is given in \eqref{vcon}.  Under Assumption \ref{lle2}, we have
\begin{equation}\label{sha}
-W_{f\!f}^{-1}W_{fl} s_L(t) =  s_F(t).
\end{equation}

\begin{figure}[t]
\centering
\includegraphics[width=1\linewidth]{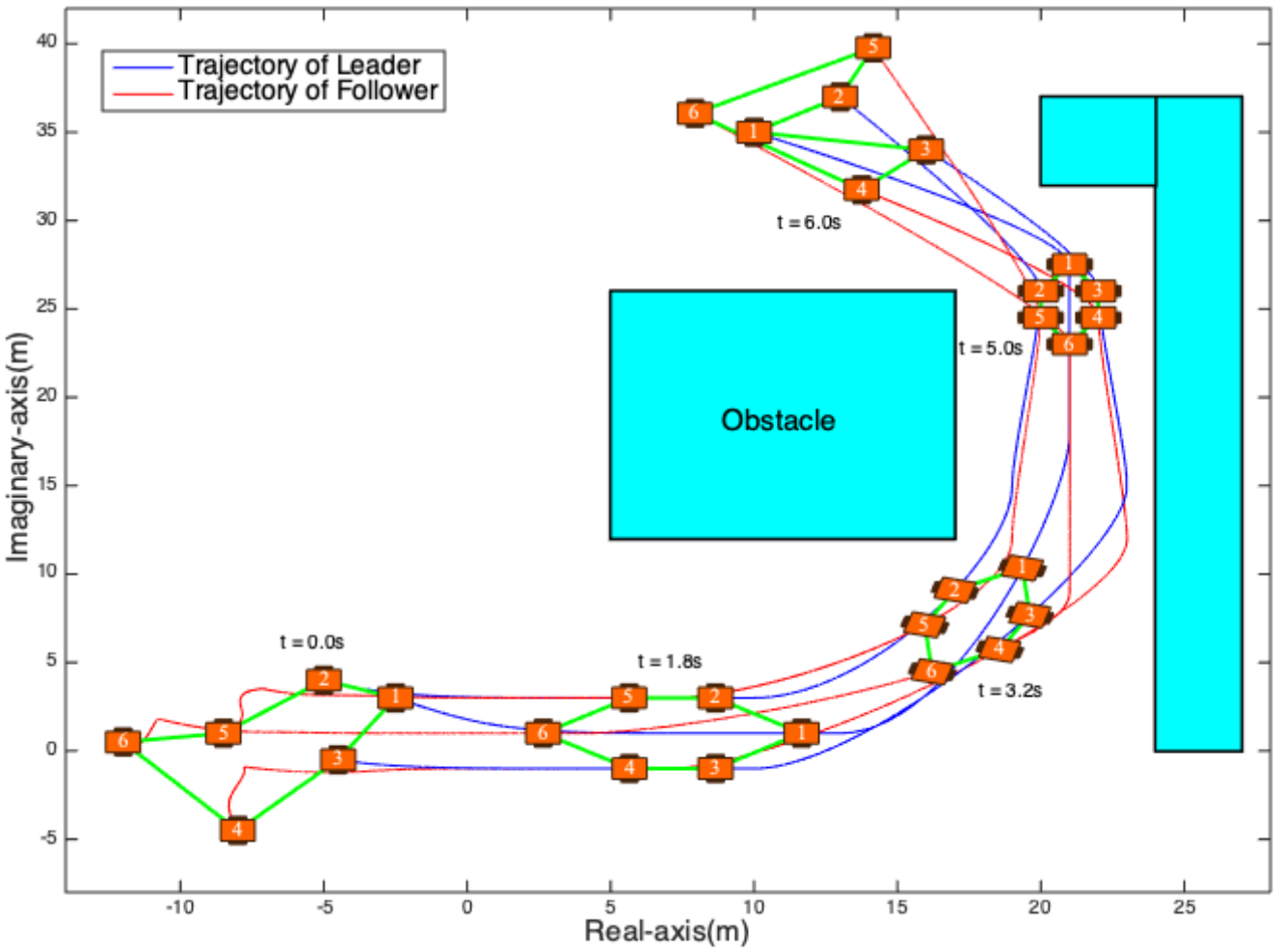}
\caption{Formation maneuver trajectories in 2-D plane.}
\label{2d1}
\end{figure}

\begin{lemma}
Suppose Assumptions \ref{lle1}-\ref{lle2} hold. Under the proposed controller \eqref{ff1}, \eqref{ff3}, or \eqref{vif23}, the followers can be controlled to form a desired shape $s_F(t)$ by only tuning the positions of the leaders if 
$\text{rank}(-W_{f\!f}^{-1}W_{fl}) \!=\! \text{rank} ([-W_{f\!f}^{-1}W_{fl} \ s_F(t)]) \le n\!-\!m$. In addition, the followers can be controlled to form any desired shape $s_F(t)$ by only tuning the positions of the leaders if the constant matrix $W_{fl}$ is full row rank. 
\end{lemma}

\begin{proof}

The complex matrices $W_{f\!f},W_{fl}$
are determined by $(\mathcal{G}, r)$. From the matrix theory, there must be a solution $s_L(t)$ to \eqref{sha} if $\text{rank}(-W_{f\!f}^{-1}W_{fl}) \!=\! \text{rank} ([-W_{f\!f}^{-1}W_{fl} \ s_F(t)]) \le n\!-\!m$. That is, under the proposed controller \eqref{ff1} or \eqref{ff3},
the followers can be controlled to form a desired shape $s_F(t)$ by only tuning the positions of the leaders.
Under Assumption \ref{lle2}, we have $\text{rank}(W_{f\!f}^{-1}W_{fl}) \!=\! \text{rank} (W_{fl})$. If the matrix $W_{fl}$ is full row rank, for any given desired shape $s_F(t)$, there must be a solution $s_L(t)$ to \eqref{sha}. Thus, the followers can be controlled to form any desired shape $s_F(t)$ by only tuning the positions of the leaders if the matrix $W_{fl}$ is full row rank. 
\end{proof}

\begin{figure}[t]
\centering
\includegraphics[width=1\linewidth]{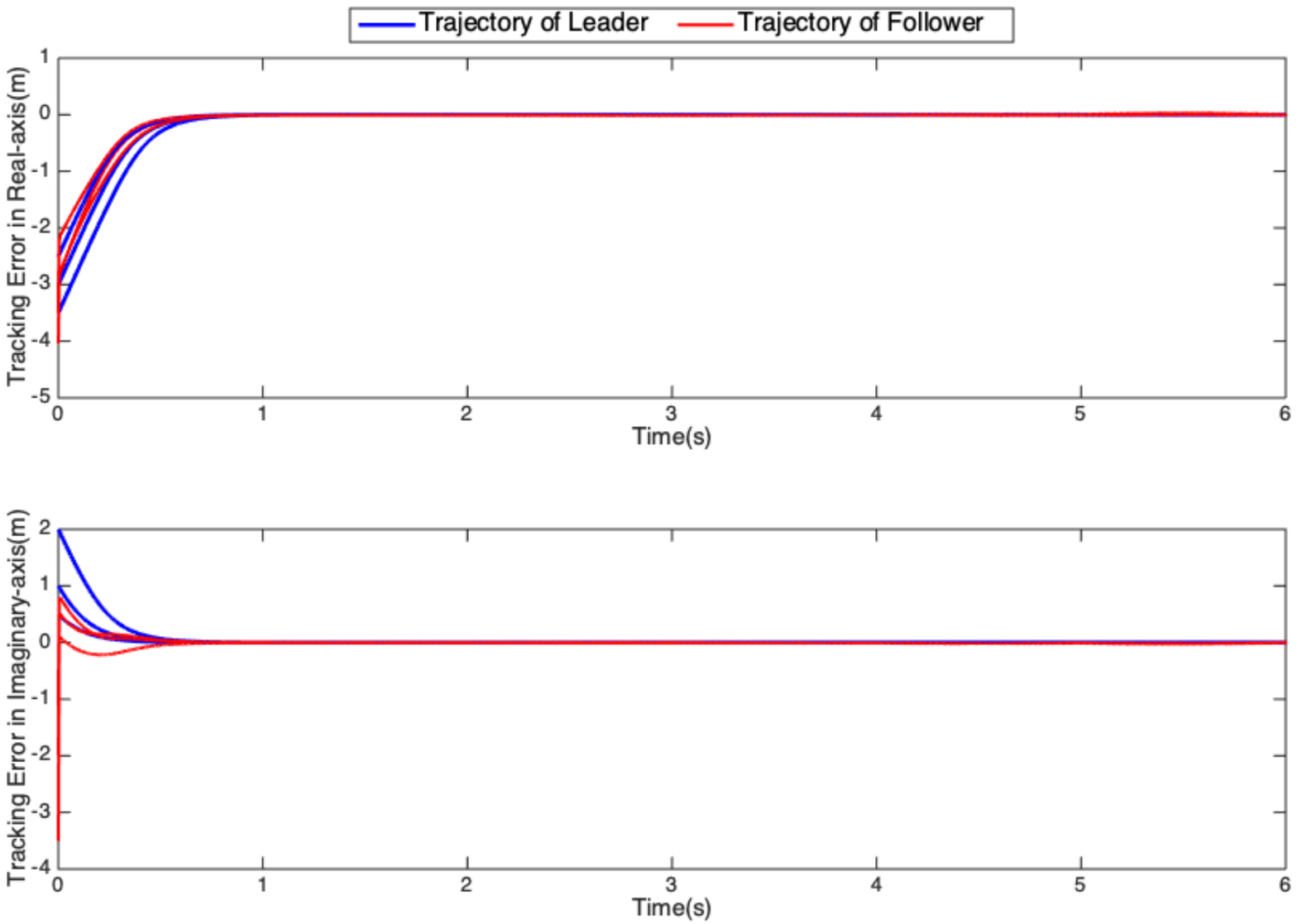}
\caption{Tracking errors in real-axis and imaginary-axis. }
\label{2d2}
\end{figure}

\begin{figure}[t]
\centering
\includegraphics[width=0.8\linewidth]{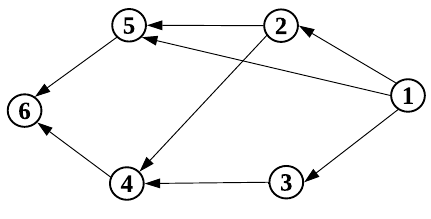}
\caption{Graph $\mathcal{G}$ of both 2-D and 3-D simulation examples. }
\label{2d}
\end{figure}

\begin{figure}[t]
\centering
\includegraphics[width=1\linewidth]{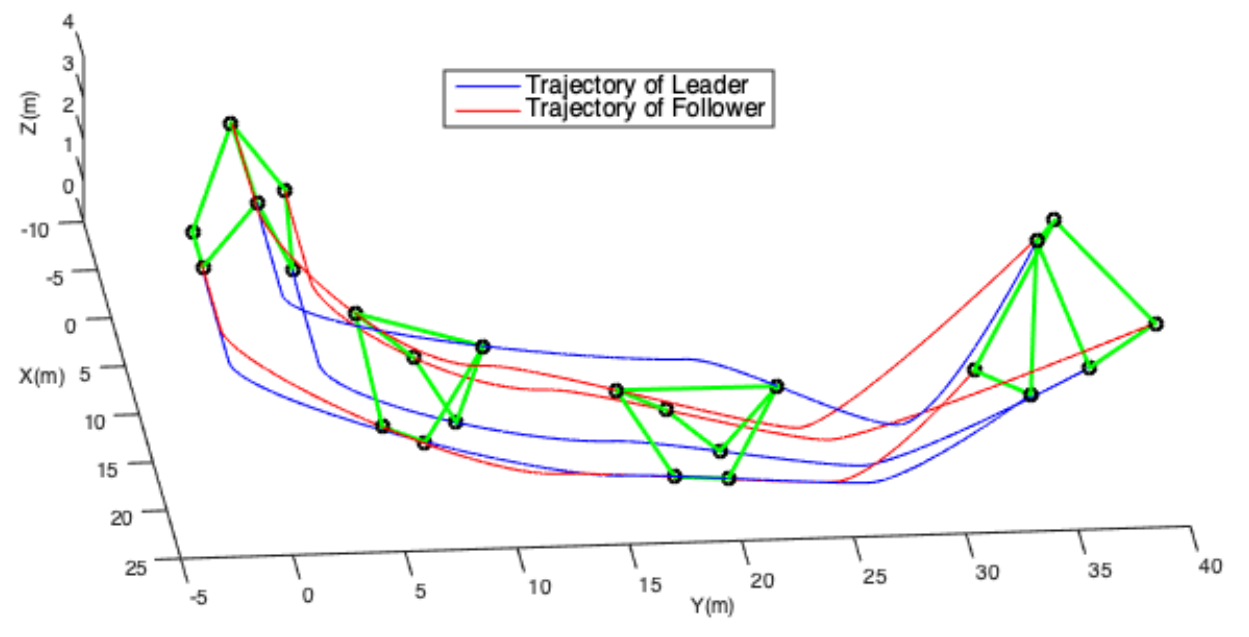}
\caption{Formation maneuver trajectories in 3-D space.}
\label{3d1}
\end{figure}

Different from the existing real-Laplacian-based approaches \cite{zhao2018affine, chen2020distributed, li2020layered, xu2020affine, onuoha2019,han2017fobarycentric,han2015three,fang2021distributed} which 
need convex, generic, or rigid nominal configurations. The proposed method can be applied in non-convex and non-generic configurations. We only require that each nominal follower and its neighbors are not collocated in each subsystem \eqref{3dm12}, e.g., $r_i \! \neq \! r_j \! \neq \! r_k$ in 2-D plane and $r_i \! \neq \! r_j \! \neq \! r_k, \epsilon_i \! \neq \! \epsilon_j \! \neq \! \epsilon_k$ in 3-D space.  In addition, more formation shapes can be realized. For example, in the real-Laplacian-based methods \cite{zhao2018affine, chen2020distributed, li2020layered, xu2020affine, onuoha2019,han2017fobarycentric,han2015three,fang2021distributed}, 
if the leaders are colinear or coplanar, the followers will be colinear or coplanar with the leaders. But in the proposed complex-Laplacian-based approach, the followers can be non-colinear or non-coplanar with the leaders even if the leaders are colinear or coplanar.

\begin{figure}[t]
\centering
\includegraphics[width=1\linewidth]{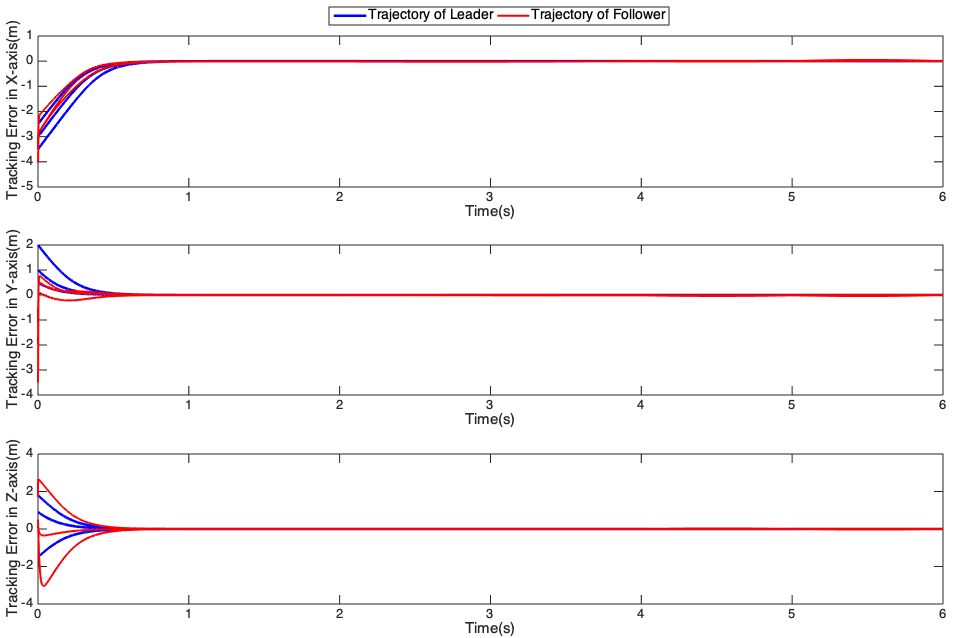}
\caption{Tracking errors in 3-D space. }
\label{3d2}
\end{figure}

\section{Simulation}\label{lation}

\subsection{Communication Graph and Nominal Configuration}

We will present
two simulation examples in both 2-D and 3-D spaces.
The multi-agent systems consist of 
three leaders $\mathcal{V}_l \!=\! \{ 1,2,3 \}$ and three followers $\mathcal{V}_f \!=\! \{ 4,5,6 \}$ shown in Fig. \ref{2d1} and Fig. \ref{3d1}, where the trajectories of the leaders and followers are marked in blue and red, respectively.  The graphs $\mathcal{G}$ used in
 both 2-D and 3-D simulation examples are the same given in Fig. \ref{2d}.  
The nominal configurations $r \!=\! [r_1^H, \cdots, r_{6}^H]^H$ and $ \epsilon \!=\! [\epsilon _1, \cdots, \epsilon _{6}]^T$ are  designed as
\begin{align}\label{sj1}
    &  \begin{array}{ll}
     & \ \ \ \ \ \ r_1 = 1+\iota, \ \ \ \ r_2 = -2+3\iota, \ \ r_3 = -2-\iota, \ \  \\
     & \ \ \ \ \ \ r_4 = -5-\iota, \ \ r_5 = -5+3\iota, \ \ r_6= -8+\iota, \\
   &  \epsilon_1 = 3.5, \epsilon_2 = 1.1,  \epsilon_3 = 1.2, 
   \epsilon_4 = 1.3,  \epsilon_5 = 2.4,  \epsilon_6= 3.5.
\end{array} 
\end{align}

The nominal configuration $r$ is used in 2-D plane, while the nominal configuration $[r, \epsilon]$ is used in 3-D space.

\subsection{Analysis and Comparison of Simulation Results}

The first simulation example shown in Fig. \ref{2d1} demonstrates 2-D formation maneuver control. It is clear from Fig. \ref{2d1} that the multi-agent system keeps maneuvering to change its translation from $0\mathbf{s}$ to $6\mathbf{s}$, rotation from $2\mathbf{s}$ to $6\mathbf{s}$, scaling from $4\mathbf{s}$ to $5\mathbf{s}$, and shape of formation from $5\mathbf{s}$ to $6\mathbf{s}$ in order to avoid obstacles and pass through the passage by tuning the  maneuver parameters, where the inter-agent collision avoidance is also guaranteed.
For example, the formation shape is changed within the time interval $t \! \in \! [5,  6]$. For the formation shape parameter $s(t) \!=\! [s_L^H(t), s_F^H(t)]^H$, its
$s_L$ is designed to smoothly interpolate between the old
shape $r$ at $t=5$ and the new shape given by \eqref{s6} at $t=6$. Then, its $s_F$ is
computed from $s_L$ by \eqref{sha}. It is worth noting that the
new shape cannot be obtained from the old shape by means of translation, scaling, or rotation. 

\begin{equation}\label{s6}
\begin{array}{ll}
     &  s(6) \!=\! [5 \!+\! \frac{35\iota}{2}; \frac{13\iota}{2} \!+\! \frac{37\iota}{2};8\!+\!17\iota; \\
     & \ \ \ \ \ \ \ \ \frac{55}{8}\!+\!\frac{127\iota}{8}; \frac{92}{13}\!+\!\frac{517\iota}{26};\frac{127}{32}\!+\!\frac{577\iota}{32}].
\end{array}
\end{equation}

Note from Fig. \ref{2d2} that 
the tracking errors in real-axis and imaginary-axis remain zero when the multi-agent system maneuvers. It is clear from Fig. \ref{2d2} that
the followers and leaders converge to their target positions at almost the same time instant. The reason is that the followers do not know their target positions and need to follow the leaders to reach their target positions.  The second simulation example 
shown in Fig. \ref{3d1} demonstrates the formation maneuver control in 3-D space.
Similar to the 2-D case, the 3-D multi-agent system in Fig. \ref{3d1} keeps maneuvering to change its translation, rotation, scaling, and shape of formation, where the agents are not required to be coplanar. The tracking errors of the follower group and leader group in each axis also converge to zero at almost the same time instant shown in Fig. \ref{3d2}.

Compared with the existing real-Laplacian-based  methods \cite{zhao2018affine, chen2020distributed,  li2020layered, xu2020affine, onuoha2019,han2017fobarycentric,han2015three,fang2021distributed} that each follower needs at least three neighbors in 2-D plane or four neighbors in 3-D space, 
each follower in the proposed complex-Laplacian-based method is allowed to have two neighbors in both 2-D plane and 3-D space as shown in Fig. \ref{2d}. Different from the existing complex-Laplacian-based methods \cite{ lin2014distributed,han2015formation,de2021distributed,fangsubmit} that are limited to 2-D plane, the proposed complex-Laplacian-based method can be applied in 3-D space as shown in Fig. \ref{3d1}. In addition, it is clear from  Fig. \ref{2d1} and Fig. \ref{3d1} that
the proposed method can change the formation shape by only tuning the positions of the leaders, which have not been explored in \cite{ lin2014distributed,han2015formation,de2021distributed,fangsubmit}. Other technical differences with existing works are discussed in Introduction and Section \ref{dis}.

\section{Conclusion}\label{conc}

This work studies the distributed formation maneuver control via complex Laplacian in both 2-D and 3-D spaces. The translation, scaling, rotation, and also the shape of formation can be changed continuously by only tuning the positions of the leaders, where the followers have no information of the time-varying target formation.
The proposed distributed control protocol can drive the agents to their target positions. In addition, the inter-agent collision can be avoided by tuning the time-varying maneuver parameters.

\ifCLASSOPTIONcaptionsoff
  \newpage
\fi

\bibliographystyle{IEEEtran}
\bibliography{papers}

\begin{thebibliography}{10}
\providecommand{\url}[1]{#1}
\csname url@samestyle\endcsname
\providecommand{\newblock}{\relax}
\providecommand{\bibinfo}[2]{#2}
\providecommand{\BIBentrySTDinterwordspacing}{\spaceskip=0pt\relax}
\providecommand{\BIBentryALTinterwordstretchfactor}{4}
\providecommand{\BIBentryALTinterwordspacing}{\spaceskip=\fontdimen2\font plus
\BIBentryALTinterwordstretchfactor\fontdimen3\font minus
  \fontdimen4\font\relax}
\providecommand{\BIBforeignlanguage}[2]{{%
\expandafter\ifx\csname l@#1\endcsname\relax
\typeout{** WARNING: IEEEtran.bst: No hyphenation pattern has been}%
\typeout{** loaded for the language `#1'. Using the pattern for}%
\typeout{** the default language instead.}%
\else
\language=\csname l@#1\endcsname
\fi
#2}}
\providecommand{\BIBdecl}{\relax}
\BIBdecl

\bibitem{wang2019cooperative1}
Y.~Wang, Y.~Wu, and Y.~Shen, ``Cooperative tracking by multi-agent systems
  using signals of opportunity,'' \emph{IEEE Transactions on Communications},
  vol.~68, no.~1, pp. 93--105, 2019.

\bibitem{ameur2019intelligent}
C.~Ameur, S.~Faquir, and A.~Yahyaouy, ``Intelligent optimization and management
  system for renewable energy systems using multi-agent,'' \emph{IAES
  International Journal of Artificial Intelligence}, vol.~8, no.~4, pp. 1--8,
  2019.

\bibitem{duan2022distributed}
P.~Duan, L.~He, Z.~Duan, and L.~Shi, ``Distributed cooperative lqr design for
  multi-input linear systems,'' \emph{IEEE Transactions on Control of Network
  Systems}, 2022.

\bibitem{amirkhani2021consensus}
A.~Amirkhani and A.~H. Barshooi, ``Consensus in multi-agent systems: a
  review,'' \emph{Artificial Intelligence Review}, pp. 1--39, 2021.

\bibitem{zhao2018affine}
S.~Zhao, ``Affine formation maneuver control of multiagent systems,''
  \emph{IEEE Transactions on Automatic Control}, vol.~63, no.~12, pp.
  4140--4155, 2018.

\bibitem{chen2020distributed}
L.~Chen, J.~Mei, C.~Li, and G.~Ma, ``Distributed leader-follower affine
  formation maneuver control for high-order multi-agent systems,'' \emph{IEEE
  Transactions on Automatic Control}, vol.~65, no.~11, pp. 4941--4948, 2020.

\bibitem{li2020layered}
D.~Li, G.~Ma, Y.~Xu, W.~He, and S.~S. Ge, ``Layered affine formation control of
  networked uncertain systems: A fully distributed approach over directed
  graphs,'' \emph{IEEE Transactions on Cybernetics}, vol.~51, no.~12, pp.
  6119--6130, 2021.

\bibitem{xu2020affine}
Y.~Xu, S.~Zhao, D.~Luo, and Y.~You, ``Affine formation maneuver control of
  high-order multi-agent systems over directed networks,'' \emph{Automatica},
  vol. 118, p. 109004, 2020.

\bibitem{onuoha2019}
O.~Onuoha, H.~Tnunay, Z.~Li, and Z.~Ding, ``Affine formation algorithms and
  implementation based on triple-integrator dynamics,'' \emph{Unmanned
  Systems}, vol.~7, no.~01, pp. 33--45, 2019.

\bibitem{han2017fobarycentric}
T.~Han, Z.~Lin, R.~Zheng, and M.~Fu, ``A barycentric coordinate-based approach
  to formation control under directed and switching sensing graphs,''
  \emph{IEEE Transactions on Cybernetics}, vol.~48, no.~4, pp. 1202--1215,
  2017.

\bibitem{han2015three}
T.~Han, Z.~Lin, and M.~Fu, ``Three-dimensional formation merging control under
  directed and switching topologies,'' \emph{Automatica}, vol.~58, pp. 99--105,
  2015.

\bibitem{fang2021distributed}
X.~Fang, X.~Li, and L.~Xie, ``Distributed formation maneuver control of
  multiagent systems over directed graphs,'' \emph{IEEE Transactions on
  Cybernetics}, vol.~52, no.~8, pp. 8201--8212, 2022.

\bibitem{lin2014distributed}
Z.~Lin, L.~Wang, Z.~Han, and M.~Fu, ``Distributed formation control of
  multi-agent systems using complex laplacian,'' \emph{IEEE Transactions on
  Automatic Control}, vol.~59, no.~7, pp. 1765--1777, 2014.

\bibitem{han2015formation}
Z.~Han, L.~Wang, Z.~Lin, and R.~Zheng, ``Formation control with size scaling
  via a complex laplacian-based approach,'' \emph{IEEE Transactions on
  Cybernetics}, vol.~46, no.~10, pp. 2348--2359, 2015.

\bibitem{de2021distributed}
H.~G. de~Marina, ``Distributed formation maneuver control by manipulating the
  complex laplacian,'' \emph{Automatica}, vol. 132, p. 109813, 2021.

\bibitem{fangsubmit}
X.~Fang, L.~Xie, and X.~Li, ``Distributed localization in dynamic networks via
  complex laplacian,'' \emph{Automatica}, vol. 151, p. 110915, 2023.

\bibitem{paden1987calculus}
B.~Paden and S.~Sastry, ``A calculus for computing filippov's differential
  inclusion with application to the variable structure control of robot
  manipulators,'' \emph{IEEE Transactions on Circuits and Systems}, vol.~34,
  no.~1, pp. 73--82, 1987.

\bibitem{lin2016distributedad}
Z.~Lin, T.~Han, R.~Zheng, and M.~Fu, ``Distributed localization for 2-d sensor
  networks with bearing-only measurements under switching topologies,''
  \emph{IEEE Transactions on Signal Processing}, vol.~64, no.~23, pp.
  6345--6359, 2016.

\bibitem{shevitz1994lyapunov}
D.~Shevitz and B.~Paden, ``Lyapunov stability theory of nonsmooth systems,''
  \emph{IEEE Transactions on Automatic Control}, vol.~39, no.~9, pp.
  1910--1914, 1994.

\end{thebibliography}

\end{document}